\documentclass[11pt]{llncs}

\usepackage{amsfonts}
\usepackage{graphicx}    

\title{Non-Identifiable Pedigrees and a Bayesian Solution}
\author{B. Kirkpatrick}
\institute{University of British Columbia}

\begin{document}         

\maketitle
\begin{abstract}
  Some methods aim to correct or test for relationships or to reconstruct the pedigree, or family tree.  We show that these methods cannot resolve ties for correct relationships due to identifiability of the pedigree likelihood which is the probability of inheriting the data under the pedigree model.  This means that no likelihood-based  method can produce a correct pedigree inference with high probability.  This lack of reliability is critical both for health and forensics applications.

  ~ ~ Pedigree inference methods use a structured machine learning approach where the objective is to find the pedigree graph that maximizes the likelihood.  Known pedigrees are useful for both association and linkage analysis which aim to find the regions of the genome that are associated with the presence and absence of a particular disease.  This means that errors in pedigree prediction have dramatic effects on downstream analysis.

  ~ ~ In this paper we present the first discussion of multiple typed individuals in non-isomorphic pedigrees, $\mathcal{P}$ and $\mathcal{Q}$, where the likelihoods are non-identifiable, $Pr[G~|~\mathcal{P},\theta] = Pr[G~|~\mathcal{Q},\theta]$, for all input data $G$ and all recombination rate parameters $\theta$.  While there were previously known non-identifiable pairs, we give an example having data for multiple individuals.
  
  ~ ~ Additionally, deeper understanding of the general discrete structures driving these non-identifiability examples has been provided, as well as results to guide algorithms that wish to examine only identifiable pedigrees.  This paper introduces a general criteria for establishing whether a pair of pedigrees is non-identifiable and two easy-to-compute criteria guaranteeing identifiability. Finally, we suggest a method for dealing with non-identifiable likelihoods: use Bayes rule to obtain the posterior from the likelihood and prior.  We propose a prior guaranteeing that the posterior distinguishes all pairs of pedigrees.

 ~ ~ Shortened version published as: {\bf B. Kirkpatrick.} Non-identifiable pedigrees and a Bayesian solution. \emph{Int. Symp. on Bioinformatics Res. and Appl. (ISBRA)}, 7292:139-152 2012.

\end{abstract}

\let\thefootnote\relax\footnotetext{{\bf Keywords:} pedigree genetics, discrete probability, identifiability.}

\vspace{-0.75cm}
\section{Introduction}
\vspace{-0.25cm}
\paragraph{Motivation.}
Pedigrees are useful for disease association~\cite{Thornton2007},
linkage analysis~\cite{Abecasis2002}, and estimating recombination
rates~\cite{Coop2008}.  Most of these calculations involve the
pedigree likelihood which is formulated using probabilities for
Mendelian inheritance given a graph of the relationships.  Since the
known algorithms for computing the likelihood are exponential, there
have been many attempts to speed up the exact likelihood
calculation~\cite{Fishelson2005,Abecasis2002,Sobel1996,Geiger2009,Browning2002,McPeek2002inference,Kirkpatrick2011xx}.
Due to the running-time issue, other statistical methods have been
introduced which perform genome-wide association studies that use a faster 
correction for the relationship
structure~\cite{Bourgain2003,Thornton2007,Thornton2010}.

Pedigree reconstruction, introduced by
Thompson~\cite{thompson1985}, is very similar to methods
used for phylogenetic tree reconstruction.  The aim is to search the
space of pedigree graphs for the graph that maximizes the likelihood, 
which is the probability of the observed data being inherited on the
given pedigree graph.  However, the pedigree reconstruction problem
differs from the phylogenetic reconstruction problem in several
important ways: 1) the pedigree graph is a directed acyclic graph
whereas the phylogeny is a tree, 2) while the phylogenetic
likelihood is efficiently computed, the only known algorithms for the
pedigree likelihood are exponential, either in the number of people or
the number of sites~\cite{Lauritzen2003}, and 3) the phylogenetic
likelihood is identifiable~\cite{Thatte2010}, while we demonstrate
that the pedigree likelihood is non-identifiable for the pedigree
graph.

Whether the pedigree likelihood is identifiable for the pedigree graph
is crucial to forensics where relationship testing is performed using
the likelihood on unlinked sites~\cite{Pinto2010}.  The scenario is
that an unknown person, $a$, leaves their DNA at the crime scene, and
it is a close match to a sample, $b$, in a database.  The relationship
between $a$ and $b$ is predicted, and any relatives of $b$ who fit
the relationship type are under suspicion.  Our results indicate
that the number of people who should fall under suspicion might be
larger than previously thought.  For example, paternity and
full-sibling testing are both common and very accurate.  However,
half-sibling relationships are non-identifiable from avuncular
relationships and from grand-parental relationships with unlinked
sites.  As we will see later, for both unlinked and linked sites,
different types of cousins relationships are also non-identifiable,
even with the addition of genetic material from a third related
person.  Due to these non-identifiable relationships, a known
relationship between a third person, $c$ and $b$ is not enough
information for conviction without also checking whether there is a
perfect match between the DNA of $c$ and $a$ and whether there is
additional information.

The likelihood is
also used to correct existing pedigrees where relationships are
mis-specified~\cite{McPeek2000,Sun2002,Stankovich2005}.  Much of their
success comes from changing relationships that result in zero or
very low likelihoods.  Again, the accuracy of these methods will be
effected by the non-identifiable likelihood. 
For similar reasons, the accuracy of pedigree relationship prediction~\cite{Stankovich2005} and reconstruction methods~\cite{thompson1985,Kirkpatrick2011b} is greatly
influenced by the likelihood being non-identifiable, since these
methods rely on the likelihood or approximations of it to guide relationship prediction.

The kinship coefficient is known to be non-identifiable for the
pedigree graph~\cite{Thompson1975}.  The kinship coefficient is an
expectation over the condensed identity states which describe the
distinguishable allelic relationships between a pair of individuals.
Pinto et al.~\cite{Pinto2010} showed that there are cousin-type pairs
of pedigrees having the same kinship coefficient.  However, these
results apply only to \emph{unlinked} sites, a
special case of the \emph{linked} sites.

This work considers identifiable pedigrees on \emph{linked} sites.
Thompson~\cite{Thompson1975} provided an early discussion of
this topic.  Donnelly~\cite{Donnelly1983}
discovered that cousin-type relationships are non-identifiable if two
pedigrees have the same total number of edges separating the two
genotyped cousins from the common ancestor.  

In this paper, we make use of a method by Kirkpatrick and
Kirkpatrick~\cite{Kirkpatrick2011xx} to collapse the original hidden
states of the likelihood HMM into the combinatorially largest
partition which is still an HMM.  Using this tool-box, we are able to
show that two pedigrees are non-identifiable if and only if they have
an isomorphism between their collapsed state spaces.  We relate this
isomorphism to known results on the non-identifiability of the kinship
coefficient.  We introduce a method of removing edges from a pedigree
to obtain a minimal pedigree having the same likelihood.  We then show
that two pedigrees that have different minimal sizes must be
identifiable.  We connect this notion of removing edges to the
pruning introduced by McPeek~\cite{McPeek2002inference} which is
clearly implementable in polynomial time, and we also introduce a result
stating that pedigrees with discrete non-overlapping generations such
as those obtained from the diploid Wright-Fisher (dWF) model are always
identifiable.

We give several examples of the kinship coefficient and pedigree
likelihood being non-identifiable.  We give the only known
non-identifiability example where there are more than two individuals
with data.  Finally, we discuss a Bayesian 
method for integrating over this uncertainty.

\vspace{-0.25cm}
\section{Background}
\vspace{-0.25cm}
A \emph{pedigree graph} is a directed acyclic graph $P=(I(P),E(P))$
where the nodes are individuals and edges are parent-child
relationships directed from parent to child.  All individuals in
$I(P)$ must have either zero or two incoming edges.  If an individual
has zero incoming edges, then that individual is a \emph{founder}.
The set of founders for pedigree graph $P$ is $F(P)$.

A \emph{pedigree} is a tuple $\mathcal{P} = (P,s,\chi,\ell)$ where $P$ is
the pedigree graph,  function $s:I(P) \to \{m,f\}$ are the genders, set $\chi
\subseteq I(P)$ is the individuals of interest, and $\ell:\chi \to
\mathbb{N}$ are the \emph{names} of the individuals of interest.  If
$i \in I(P)$ has two incoming edges, $p_0(i)$ and $p_1(i)$, then one parent
must be labeled $s(p_{j}(i)) = m$ and the other $s(p_{1-j}(i)) = f$ for $j \in \{0,1\}$.

The likelihood, $Pr[G~|~\cal{P},\theta]$, is a function of the
genotypes $G$, the recombination rates $\theta$, and the pedigree
$\mathcal{P}$.  However, we will abuse notation by referring
to a pedigree by its pedigree graph and writing 
$Pr[G~|P,\theta]$.  In these instances, the set $\chi$ will be clear
from the context.  

Two pedigrees $\mathcal{P}$ and $\mathcal{Q}$ are said to be
\emph{identifiable} if and only if $Pr[G~|~\mathcal{P},\theta] \ne
Pr[G~|~\mathcal{Q},\theta]$ for some values of $G$ and $\theta$.  If
$\mathcal{P}$ and $\mathcal{Q}$ are not identifiable, we call them
\emph{non-identifiable}.

Two pedigree graphs, $P$ and $Q$ are \emph{isomorphic} if there exists
a mapping $\phi:I(P) \to I(Q)$ such that $(u,v) \in E(P)$ if and only
if $(\phi(u), \phi(v)) \in E(Q)$.  This is an isomorphism of the
pedigree graph rather than of the pedigree, because the genders are
not necessarily preserved by the map $\phi$.  From now on, we will
assume that $P$ and $Q$ are not isomorphic.

Two isomorphic pedigrees might have different gender labels,
and they would be identifiable when considering sex-chromosome data.  We
restrict our discussion to autosomal data, where these two pedigrees
would be non-identifiable.

\paragraph{The Hidden Markov Model.}
Rather than writing out the cumbersome likelihood equation, we will
define the likelihood by specifying the HMM.  For each pedigree 
$\mathcal{P}=(P,s,\chi,\ell)$, there is an HMM, and everything in this
section is defined relative to a specific pedigree $\mathcal{P}$.  To
specify the HMM, we need to specify the hidden states, the emission
probability, and the transition probabilities.  We will begin with the
hidden states.

An inheritance vector $x \in \{0,1\}^n$ has length $n = |E(P)|$.  Each
bit, $x_e$, in this vector indicates which grand-parental allele,
maternal or paternal, was inherited along edge $e \in E(P)$.  
An \emph{inheritance graph} $R_x$ 
contains two nodes for each individual in $i \in I(P)$, called $i_0$
and $i_1$, and edges $(p_{j}(i)_{x_e}, i_j)$ for each $(p_{j}(i),i)
\in E(P)$.  The sets $\chi_0$ and $\chi_1$ are the paternal and
maternal alleles, respectively, of the individuals of interest.  We
will refer to the collective set $\chi_0 \cup \chi_1$ as the
\emph{alleles of interest}.  Each node in $R_x$ represents an allele.
The inheritance graph is a forest with each root being a
founder allele.  The inheritance vectors are the \emph{hidden states} of
the HMM.  Let $\mathcal{H}_P$ be the hypercube of dimension $|E(P)|$;
its vertices represent all the inheritance vectors.

This inheritance graph represents identity-by-descent (IBD) in that
any pair of individuals of interest $i,i' \in \chi$ are IBD if there
exists an inheritance vector $x$ such that one pair of $(i_0,i'_0)$,
$(i_0,i'_1)$, $(i_1,i'_0)$ or $(i_1,i'_1)$ are connected.  The
\emph{identity states}, are the sets of the partition induced on the alleles of
interest by the connected components of $R_x$, namely 
$D_x = \{y \in \mathcal{H}_{P} | CC(R_y) = CC(R_x)\}.$
The \emph{transition probabilities} are a function of the per-site recombination rates $\theta = (\theta_1,..,\theta_{T-1})$ for $T$ sites.
Let $X_t$ be the random variable for the hidden state at site $t$.
The probability of recombining from hidden state $x$ to  state $y$ at site $t$ is
\vspace{-0.25cm}
\begin{eqnarray}
\label{Xtransition}
Pr[X_{t+1} = y~|~X_{t} = x, \theta] = \theta_t^{H(x,y)}(1-\theta_t)^{n-H(x,y)}
\vspace{-0.25cm}
\end{eqnarray}
where $H(x,y) = |x \oplus y |_1$ is the Hamming distance between the two bit vectors, $\oplus$ indicates the XOR operation, and $|.|_1$ is the $L_1$-norm.  In some instances, we may make the $\theta$ implicit, because it is clear from context.

The \emph{emission probability} depends on the data, which is the
genotype random variable $G$.  Each individual of interest $i \in
\chi$ has two rows in the genotype matrix which encode, for each
column $t$, the alleles that appear in that individual's genome.  For
example, $\{g_{it}^0,g_{it}^1\}$ from the $0$th and $1$st rows for
individual $i$ at site $t$ is the (unordered) set of alleles that
appear in that individual's genome.  The data for all the individuals
at site $t$ is an $n$-tuple $g_t = (\{g_{it}^0,g_{it}^1\} | \forall i)$ and $g =
(g_1,...,g_T)$ is the data at all $T$ sites.  The pedigree HMM
deconvolves these unordered alleles by considering all possible
orderings of the genotypes when assigning them to the hidden alleles.

Specifically, let $CC(R_x)$ be the connected components of $R_x$.
Then the emission probability at site $t$ is
\vspace{-0.25cm}
\[
  Pr[G_t = g_t~|~X_t = x, P] \propto \sum_{\tilde{g}_t} \prod_{c \in CC(R_x)} \mathbf{1}\{n(c,\tilde{g}_t) = 1\} Pr[h(c,\tilde{g}_t)]
\vspace{-0.25cm}
\]
where $\tilde{g}_t$ is the ordered alleles $(g_{it}^0,g_{it}^1)$ that
appear in $g_t$, $n(c,\tilde{g}_t)$ is the number of alleles assigned to
$c$ by $\tilde{g}_t$, and $h(c,\tilde{g}_t)$ is the allele of $\tilde{g}_t$
that appears in $c$.  Notice that by definition of the identity
states, $\{D_x | \forall x\}$, $Pr[G_t~|~X_t=x_1] = Pr[G_t~|~X_t=x_2]$ for
all $x_1,x_2 \in D_x$.

This completes the definition of the HMM and the likelihood.  Now, our
task is to find pairs of pedigree graphs $(P,Q)$ such that
$Pr[G~|~P,\theta] = Pr[G~|~Q,\theta]$ for all $G$ and $\theta$.  We
can do this by considering multiple equivalent HMMs and finding the
``optimal'' HMM that describes the likelihood of interest.  Given two
optimal HMMs, we can easily compare their likelihoods for different
values of $G$ and $\theta$.

\paragraph{The Maximum Ensemble Partition.}
In this paper, we will use a method similar to that discussed by
Browning and Browning~\cite{Browning2002} and improved by Kirkpatrick
and Kirkpatrick~\cite{Kirkpatrick2011xx}.  This method relies on an
algebraic formulation of the hidden states of the Hidden Markov Model
(HMM) that is used to compute the pedigree likelihood.  Specifically,
we can collapse the original hidden states into the combinatorially
largest partition which is still an HMM.  From the collapsed state
space (termed the maximum ensemble partition), we can easily see that
certain pairs of pedigrees have isomorphic HMMs and thus identical
likelihoods.

For pedigree $\mathcal{P} = (P,s,\chi,\ell)$, consider a new HMM with
hidden states $Y_t$ in a state space that is defined by a partition,
$m(P) := \{W_1,...,W_k\}$, of $\mathcal{H}_P$, meaning that for all
$i,j$, $W_i \cap W_j = \emptyset$ and $\cup_{i=1}^k W_i =
\mathcal{H}_P$.  For the HMM for $Y_t$ to have the same
likelihood as the HMM for $X_t$ the \emph{Markov property} and the
\emph{emission property}, defined next, must be satisfied.

Let the transition probabilities of $Y_t$ be the expectation of $X_t$ as follows, for all $i,j$, and for  $x \in W_i$
\vspace{-0.25cm}
\begin{eqnarray}
\label{Ytranstion}
Pr[Y_{t+1}=W_j~|~Y_t = W_i] &=& Pr[X_{t+1} \in W_j~|~X_t= x]  \\
&=& \sum_{y \in W_j} Pr[X_{t+1}=y~|~X_t=x].
\vspace{-0.25cm}
\end{eqnarray}
Conditioning on $\theta$ is implicit on both sides of the equation.
The \emph{Markov property} is required for $Y_t$ to be Markovian:
\vspace{-0.25cm}
\[
\sum_{y \in W_j} Pr[X_{t+1}=y~|~X_t=x_1] = \sum_{y \in W_j}  Pr[X_{t+1}=y~|~X_t=x_2]
\vspace{-0.25cm}
\]
for all $x_1, x_2 \in W_i$ for all $i$ and for all $W_j$. For more
details, see~\cite{Browning2002,Kirkpatrick2011xx}.

The \emph{emission property} states that the emission probabilities of
$X_t$ impose a constraint on $Y_t$.  This constraint is that the
partition $\{W_1,...,W_k\}$ must be a sub-partition of the partition
induced on the hidden states by the emission probabilities: \\
$E_x(P) = \left\{y \in \mathcal{H}_P ~|~ Pr[G_t=g_t~|~X_t=x] = Pr[G_t=g_t~|~X_t=y] ~\forall g_t \right\}.$\\
We call the set $\{E_x(P) | \forall x\}$ the \emph{emission partition} since it partitions the state-space $\mathcal{H}_P$.

It has been shown in~\cite{Kirkpatrick2011xx} that the partition
$\{W_1,...,W_k\}$ which satisfies the Markov property and the emission
property and which maximizes the sizes of the sets in the
partition---i.e.~$\max_{i \in \{1,...,k\}} |W_i|$---can be found in
time $O(nk2^n)$ where $n$ is the number of edges, and $k$ is a
function of the known symmetries of the pedigree graph $k \le 2^n$.
We call this partition the \emph{maximum ensemble partition}.


It turns out that the maximum ensemble partition is unique, making the
derived HMM the unique ``optimal'' representation for the likelihood.
We will exploit this fact to find non-identifiable pairs of pedigrees.

\vspace{-0.25cm}
\section{Methods}
\vspace{-0.25cm} We will define a general criteria under which a pair
of non-isomorphic pedigree graphs have identical likelihoods for all
input data and recombination rates, as wells as define a
uni-directional polynomial-checkable criteria whereby we can determine
whether some pairs of pedigrees are identifiable.  In the following
section, we will apply these results to investigate when pedigrees are
identifiable, to give several examples where the pedigrees are
non-identifiable, and to suggest a Bayesian solution.

Given two non-isomorphic pedigree graphs $P$ and $Q$, and their maximum ensemble partitions $m(P)$ and $m(Q)$, respectively.  We say that $\psi$ is a \emph{proper isomorphism} if $\psi$ is a bijection $m(P)$ onto $m(Q)$  such that the following hold:
\begin{description}
\vspace{-0.25cm}
\item[Transition Equality]  $Pr[Y_{t+1}^P~|~Y_t^P,\theta] = Pr[\psi(Y_{t+1}^P)~|~\psi(Y_{t}^P),\theta] ~~\forall t$
\item[Emission Equality]    $Pr[G_t~|~Y_t^P,P] = Pr[G_t~|~\psi(Y_t^P), Q] ~~\forall t$
\vspace{-0.25cm}
\end{description}
where $Y_{t}^P$ is the random variable for the hidden state for pedigree $P$.

\begin{theorem}
\label{thm:equivalence}
There exists isomorphism $\psi:m(P) \to m(Q)$ satisfying the transition and emission equalities if and only if the likelihoods for $\mathcal{P}$ and $\mathcal{Q}$ are non-identifiable, $Pr[G~|~\theta, P] = Pr[G~|~\theta, Q]$, for all $G$ and $\theta = (\theta_1,...,\theta_{T-1})$ where $T$ is the number of sites and $T \ge 2$.
\end{theorem}
\begin{proof}
($\Rightarrow$) Given a proper isomorphism $\psi:m(P) \to m(Q)$ that satisfies the transition and emission equalities, the likelihoods are necessarily the same, by definition of the Hidden Markov Model.

($\Leftarrow$) Given that the two pedigrees are identifiable, we will construct $\psi$.  Consider pedigrees $P$ and $Q$.  They both have unique maximum ensemble partitions $m(Q)$ and $m(P)$~\cite{Kirkpatrick2011xx}.  By the definition of $Pr[G~|~\theta, Q]$, this distribution can be represented by an HMM, called $\mathcal{M}(Q)$, over state-space $m(Q)$.  By the equality $Pr[G~|~\theta, P] = Pr[G~|~\theta, Q]$, we know that there is an HMM for $P$, $\mathcal{M}(P)$, with the same transition matrix and emission probabilities as $\mathcal{M}(Q)$.  Since $\mathcal{M}(Q)$ has maximum ensemble state-space $m(Q)$, then by uniqueness, there is no other state-space that is as small.  By the equality of the two distributions, we know that $\mathcal{M}(P)$ also has maximum ensemble state-space $m(P)$.  But since $m(P)$ is the unique maximum ensemble state-space for $\mathcal{M}(P)$, there must be an isomorphism $\psi:m(P) \to m(Q)$ satisfying the transition and emission equalities. \qed
\end{proof}

To apply this method, we need to obtain $m(P)$ and $m(Q)$ and the appropriate proper isomorphism $\psi$.  To obtain $m(P)$ and $m(Q)$ we rely on the maximum ensemble algorithm~\cite{Kirkpatrick2011xx}.  The proper isomorphism is obtained by examining the transition probabilities of the respective HMMs.

\begin{corollary}
\label{cor:unlinked}
For unlinked sites $\theta_t = 0.5$ for all $1 \le t \le T-1$, for any pedigree graphs $P$ and $Q$ with maximum ensemble states $|m(P)| = |m(Q)|$ and identical identity states, the pedigrees are non-identifiable. (proven in Appendix)
\end{corollary}

We are now in a position to relate non-identifiability on pedigree
HMMs to non-identifiability of an important calculation that relies
on independent sites---the kinship coefficient.  The \emph{kinship
  coefficient} for a pair of individuals of interest is defined as the
probability of IBD when randomly choosing one allele from each
individual of interest.  Let the two individuals of interest be $\chi
= \{a,b\}$.  We write the kinship coefficient for $\chi$ as
$\Phi_{I}(P)_{\chi} = \sum_{x} \frac{\eta(x,\chi)}{4} \frac{1}{2^n}$
where $\eta(x,\chi)$ is the number of pairs of alleles of interest
$\chi_0 \cup \chi_1$ sharing the same connected component in $R_x$ and
$\chi_0 \cup \chi_1 = \{\{a_0,b_0\}, \{a_0,b_1\}, \{a_1,b_0\},
\{a_1,b_1\}\}$.  


\begin{corollary}
\label{cor:kinship}
For unlinked sites $\theta_t = 0.5$ for all $1 \le t \le T-1$, given two 
non-identifiable pedigree graphs, $P$ and $Q$, with two
individuals of interest $\chi=\{a,b\}$, the kinship coefficient is
identical. (proven in Appendix)
\end{corollary}

This last corollary is a uni-directional implication.  There are some pairs of pedigrees $P$ and $Q$ for which the kinship coefficient is equal but for which the likelihood is identifiable, see Fig~\ref{fig:example1}.

The final set of results we introduce will try to answer the question of when are pedigrees identifiable.  Since some algorithms use the likelihood to choose the best pedigree graph or relationship type, these results give some guarantees for when those algorithms will make correct decisions.
We wish to show that under some definition of ``necessary'' edges for some individuals of interest, pedigrees $P$ and $Q$ with different numbers of necessary edges have no proper isomorphism and are, therefore, identifiable. We will relate our definition of a necessary edge to the literature.
And, we will establish an even more restricted class of pedigrees for which no pair of pedigrees is identifiable.  This is the class of all dWF pedigrees.

For an edge, $e$, in the pedigree, let $\sigma$ be the indicator vector with bits $\sigma_f = 0$ for all $f \ne e$ and $\sigma_e = 1$.
For pedigree $P$ having states $\{W_1,...,W_k\}$, we will define an edge $e \in E(P)$ to be \emph{superfluous} if and only if the following two properties hold
\begin{description}
\vspace{-0.25cm}
\item[1)] $Pr[X_{t+1} = y | X_t = x] = Pr[X_{t+1} = \sigma \oplus y | X_t = \sigma \oplus x]$, for every $y \in W_j$ and $x \in W_i$ and for ever $i$ and $j$, and
\item[2)]  $Pr[G_{t}| X_t = x] = Pr[G_{t} | X_t = \sigma \oplus x]$ for all $x \in \mathcal{H}_P$.
\vspace{-0.25cm}
\end{description}
Conversely, an edge $e$ is \emph{necessary} if it is not superfluous.  For an example, see the edge adjacent to the grand-father in $P'$ of Fig~\ref{fig:example3}.

\begin{lemma}
We say that an edge is \emph{removed} if its bit is set to a fixed value in all the inheritance vectors.
Any superfluous edge can be removed without changing the value of the likelihood. (proven in Appendix)
\end{lemma}

\begin{theorem}
\label{thm:cardinalityofedges}
If pedigrees $P$ and $Q$ have a different number of \emph{necessary} edges, then there is no proper isomorphism and the likelihoods for $P$ and $Q$ are identifiable. (proven in Appendix)
\end{theorem}

In order to connect our definition of superfluous edges to the literature we will reiterate McPeek's formulation of \emph{superfluous} individuals~\cite{McPeek2002inference}.  An \emph{individual $i \in I(P)$ is superfluous} if for every pair $\{a,b\} \in \chi$ at least one of the following holds:
\vspace{-0.25cm}
\begin{enumerate}
 \item $i \notin A(a) \cup A(b)$ where $A(a)$ is the ancestors of $a$
 \item $A(i) \cap \{a,b\} = \emptyset$ and there exists some $c \in I(P) \setminus \{a,b\}$ and $d \in I(P)$ such that for every $e \in \{i\} \cup A(i)$ for every $l \ge 1$ and every directed path $q = (q_0,...,q_l)$ of length $l$ with $q_0 = e$ and $q_l \in \{a,b\}$, we have c = $q_m$ and $d = q_{m+1}$ for some $0 \le m \le l-1$.
\vspace{-0.25cm}
\end{enumerate}
This last condition states that every directed path from $i$ or an ancestor of $i$ to $\{a,b\}$ must pass through directed edge $(c,d)$.  

The reason for the definition of superfluous individuals is that it is polynomial-time checkable.  If one were to directly check the definition of superfluous edges, one would find it necessary to compute the emission partition and the maximal ensemble state space which requires exponential time.  Despite this, from the definition of superfluous edges, it is easy to see the operational consequence: edges can be removed from the pedigree.  Superfluous edges and superfluous individuals are related as follows.

\begin{lemma}
An individual is \emph{superfluous} if and only if all the edges adjacent to that individuals are \emph{superfluous}. (proven in Appendix)
\end{lemma}

Theorem~\ref{thm:cardinalityofedges} tells us that when two pedigrees have a different number of necessary edges they are certainly identifiable.  While this criteria is useful if we are interested in a particular pedigree, it does not allow us to draw broad conclusions about a class of pedigrees.  Ideally, if we want to integrate over the space of pedigrees, we would want to integrate only over identifiable pedigrees for efficiency of computation.

The class of diploid Wright-Fisher (dWF) pedigrees are haploid Wright-Fisher genealogies which are two-colorable where there is a color for each gender.  These pedigrees have discrete non-overlapping generations, and all the individuals of interest are `leaves' of the genealogy.

\begin{theorem}
\label{thm:dwf}
Two non-isomorphic, dWF pedigrees $P$ and $Q$ contain only necessary edges and have individuals of interest $\chi$ labeling the `leafs' which are the individuals with no children.
Then pedigrees $P$ and $Q$ are identifiable. (proven in Appendix)
\end{theorem}

\vspace{-0.25cm}
\section{Examples}
\vspace{-0.25cm}
We will consider several examples.  The first of which is a trio of pedigrees that are non-identifiable with data from unlinked sites.  This fact is well known due to their identical kinship coefficients.  However, these three pedigrees are identifiable with data from linked sites.  The second example is an extension of the well-known non-identifiable cousin-type relationships.  In this example, we extend the relationship from two to three individuals of interest and show that the relationships remain non-identifiable.  To the best of our knowledge, this is the first example of non-identifiable pedigrees on more than two individuals of interest.

\begin{figure}[ht!]
  \begin{center}
    \includegraphics[scale=0.7]{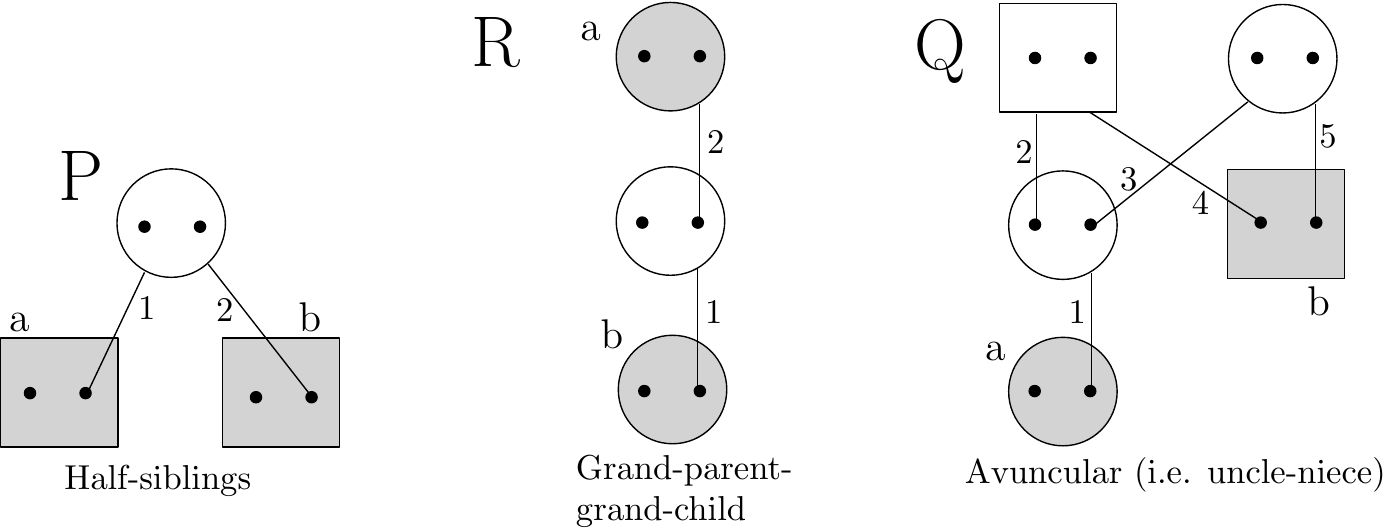}
  \end{center}
  \caption{{\bf Half-siblings, grand-parent-grand-child, and avuncular relationships are identifiable.} Individuals are drawn as boxes, if male, and circles, if female.  The individuals of interest are $\chi=\{a,b\}$.  Alleles are drawn as disks with a line between the allele and the parent it was inherited from.  For each edge, numbered $e \in \{1,...,5\}$, the binary value $x_e$ in the inheritance vector indicates which parental allele was chosen for that hidden state where zero indicates the paternal and the leftmost of the two alleles.  The numbers labeling the edges indicate in which order the bits appear in their respective vectors.  These three relationships have identical kinship coefficients.  The likelihoods of these relationships are identifiable given data on linked sites.
}
  \label{fig:example1}
\end{figure}

\paragraph{Half-Sibling, Avuncular, and Grandparent-Grandchild Relationships.}
The first example we will consider is the well-known trio of pedigrees where the kinship coefficient is identical: half-sibling, avuncular, and grandparent-grandchild relationships.  There are two individuals of interest, $a$ and $b$ for whom we have data.  These three relationships are drawn in Fig~\ref{fig:example1}.

The maximum ensemble partition for each of these three pedigrees are $\{ W_1^P =\{00,11\}, W_2^P = \{01,10\}\}$ for the half-siblings, $\{W_1^R = \{00,01\},$ $W_2^R = \{10,11\}\}$ for the grand-parent-grand-child, and for the avuncular relationship:
\vspace{-0.25cm}
\begin{eqnarray*}
W_1^Q = \{00000, 01010, 00101, 01111, 10000, 11010, 10101, 11111\} \\
W_2^Q = \{00001, 01011, 00100, 01110, 10010, 11000, 10111, 11101\} \\
W_3^Q = \{00010, 00111, 01000,  01101, 10001, 10100, 11011, 11110\} \\
W_4^Q = \{00011, 00110, 01001, 01100, 10011, 10110, 11001, 11100\}
\vspace{-0.25cm}
\end{eqnarray*}

To get the transition probabilities, we need to sum Equation~\ref{Xtransition} as in Equation~\ref{Ytranstion}.  Since for the first two pedigrees, $P$ and $R$, there are only two states, we need only compute the transition probability for one state (the others are obtained by observing that the transition probabilities sum to one).  For pedigree $P$, we have 
\vspace{-0.25cm}
\[ Pr[Y_{t+1}^P = W_1^P~|~Y_t^P = W_1^p] = (1-\theta_t)^2 + \theta_t^2 = 2\theta_t^2 - 2\theta_t +1.
\vspace{-0.25cm}
\]
For pedigree $R$, 
\vspace{-0.25cm}
\[ Pr[Y_{t+1}^R = W_1^R~|~Y_t^R = W_1^R] = (1-\theta_t)^2 + \theta_t(1-\theta_t) = 1-\theta_t.
\vspace{-0.25cm}\]

It is evident that there is no proper isomorphism that has transition equality for pedigrees $P$ and $R$.  For pairs $P,Q$ and $R,Q$ there is no proper isomorphism, because all three pedigrees contain only necessary edges and both $|m(P)| \ne |m(Q)|$ and $|m(R)| \ne |m(Q)|$.
So, we conclude that these pedigrees are identifiable as long as the number of sites $T \ge 2$ and $\theta_t < 0.5$ for all $1 \le t \le T-1$.
Despite the well-known fact that these three pedigrees have identical kinship coefficients, these pedigrees \emph{are} identifiable when the data is from multiple linked sites.  To the best of our knowledge, this paper is the first to prove this simple fact.


\begin{figure}[ht!]
  \begin{center}
    \includegraphics[scale=0.6]{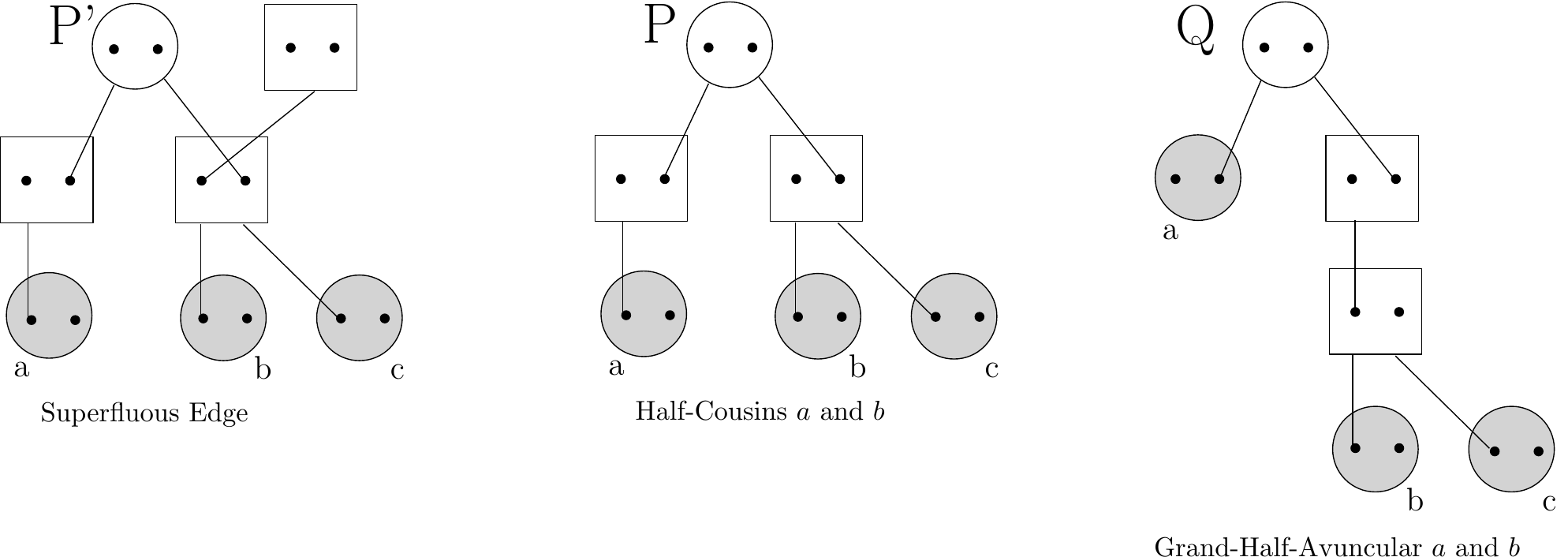}
  \end{center}
  \caption{{\bf Half-cousins and grand-half-avuncular relationships are \emph{non-}identifiable even when there is a third individual of interest.} Pedigree $P$ is derived from pedigree $P'$ by removing the superfluous edge.  The two pedigree graphs, $P$ and $Q$ are not isomorphic, yet the likelihoods are non-identifiable, meaning that no amount of data on the individual $a,b$, and $c$ will distinguish these likelihoods.}
  \label{fig:example3}
\end{figure}

\paragraph{Half-Cousins and Full-Cousins Relationships.}

To the best of our knowledge Donnelly~\cite{Donnelly1983} was the first to remark that pairs of pedigrees either of the half-cousin or of the full-cousin type and having equal numbers of edges are non-identifiable.  
Figure 6 of~\cite{Donnelly1983}
illustrates this situation.  Suppose we have two pedigrees $P_{d_a,d_b}$ and $P_{d'_a,d'_b}$ each having two individuals of interest, $\chi=\{a,b\}$ at the leaves, and the most recent common ancestors of $\chi$ have the same relationship type in both pedigrees, either half or full relationships.
Let $d_a$ and $d_b$ be the number of edges or meioses that separate individuals $a$ and $b$ from their common ancestor(s) in pedigree $P_{d_a,d_b}$.   Then as long as $d_a + d_b = d'_a + d'_b$, the two pedigrees are non-identifiable.  

Donnelly remarked this means that no amount of autosomal genetic information can distinguish these two pedigrees, ``unless of course information is available on a third person related to both of the individuals in question.''  Figure~\ref{fig:example3} shows that for some third individuals these relationships remain non-identifiable.  To the best of our knowledge, this is the first example of a pair of non-identifiable pedigrees each having three individuals of interest.

By Theorem~\ref{thm:equivalence} and Corollary~\ref{cor:kinship} we can show that \emph{both} the pedigree likelihood and the kinship coefficient are non-identifiable for half-cousin-type relationships, see Figure~\ref{fig:example3}.  The isomorphism is omitted for space reasons.  We believe that a similar result can be obtained for the full-cousin-type relationship.  However, the number of edges is large enough that calculation is difficult due to the exponential algorithm.

These examples mean that the likelihood alone is not a practical tool for testing relationships, for inferring pedigrees, or for correcting pedigrees that have relationship errors since the pedigrees under consideration might be non-identifiable.

\vspace{-0.25cm}
\section{A Potential Solution}
\vspace{-0.25cm}
This paper has focused on the likelihood  $Pr[G|P,\theta]$,
since it is currently the object being used for relationship testing and pedigree reconstruction.  However, a common alternative to the likelihood is the posterior distribution obtained via Bayes rule
\vspace{-0.25cm}
\[
Pr[P|G, \theta] = \frac{Pr[G|P,\theta]Pr[P]}{Pr[G|\theta]} = \frac{Pr[G|P,\theta]Pr[P]}{\sum_Q Pr[G|Q,\theta]Pr[Q]}.
\vspace{-0.25cm}
\]
The utility of this expression is that the posterior $Pr[P|G, \theta]$  will distinguish between non-identifiable pedigrees provided that the prior has the property that $Pr[P] \ne Pr[Q]$ when $P$ and $Q$ are non-identifiable.  Indeed, the uniform distribution over dWF pedigrees is such a prior.  Taking care with the zero-probability pedigrees which do not occur under the dWF model, we suggest a refinement.  Let $W$ be the set of all dWF pedigrees, and let $\bar{W}$ be the pedigrees which are not dWF.  Then, let $Pr[P] = 1/(|W|+1)$ for $P\in W$, and for an arbitrary ordering $Q_1,...,Q_{|\bar{W}|}$ with $Q_i \in \bar{W}$, let $Pr[Q_i] = (1/2^i)/ (Z(|W|+1))$  where $Z = \sum_{i=1} 1/2^i$.  Since the number of non-diploid WF pedigrees are countably infinite, we can approximate $Z$ using its limit $Z=1$.

Now that we have a prior, the challenge of using the posterior is that the partition function, the denominator $Pr[G|\theta]$, is most certainly intractable.  This is because there are an exponential number of pedigrees and the likelihood algorithm has exponential run-time for each pedigree.  

The intractability of the partition function points to the use of sampling methods, in particular, the Metropolis-Hastings Markov Chain Monte Carlo approach might be well suited to this problem.
Indeed, MCMC facilitates computing the proposed prior, because we can simply take the $Q_i$ in the order that they are encountered by the Markov chain.
  If we obtain a sample pedigree $P^\tau$, we can draw a new pedigree $P^{\tau+1}$ by proposing a pedigree $Q$ according to a proposal distribution $q[Q|P^\tau]$ and then choosing to accept, $P^{\tau+1} = Q$ with probability
\vspace{-0.25cm}
\[
  \min \left\{1,  \frac{Pr[G|Q,\theta]Pr[Q]}{Pr[G|P^\tau,\theta]Pr[P^\tau]} \frac{q[P^\tau|Q]}{q[Q|P^\tau]}\right\}
\vspace{-0.25cm}\]
otherwise $P^{\tau+1} = P^\tau$ remains unchanged.  A sequence of $P^1, P^2,...,P^\tau$ is guaranteed to converge to the stationary distribution $Pr[P^\tau|G,\theta]$.  After convergence at time-step $\tau$, take $\delta$ pedigree samples 
$\{P^{\tau},P^{k+\tau},...,P^{\delta k+\tau}\}$ where $k$ is the number of steps between samples.  Those samples can yield information about the posterior distribution, such as the confidence for each edge.  One could also take the most probable pedigree that was sampled, and treat that as the estimated pedigree.

The complexity here comes down to three issues, first the likelihood calculation which is exponential, second the prior on the pedigrees which might be tailored to a specific set of pedigrees having positive probabilities, i.e. those containing particular ``known'' edges, and third calculating the proposal distribution which should be tractable and produce non-zero pedigrees.  The latter is critical, because MCMC methods will not converge if they repeatedly propose zero-probability events.  This can probably be overcome by using moves inspired by the phylogenetic prune and re-graft method.  As yet, all these details are an open problem.

Alternative to integration over the whole space of pedigrees, if we have a single pedigree of which we are fairly confident, we could use this method to integrate over `nearby' pedigree graphs to get a measure of our confidence in our chosen pedigree.  We could use Theorem~\ref{thm:cardinalityofedges} as a guide to integrate only over a set of pedigrees all having the same number of necessary edges while giving a zero prior to all other pedigrees.  Such an approach might even be computationally feasible due to the polynomial-time checkable definition of necessary edges.  This would allow us to incorporate into our calculations the uncertainty we have about our chosen pedigree relative to its non-identifiable `neighbors'.

\vspace{-0.25cm}
\section{Discussion}
\vspace{-0.25cm}
This paper reviews the pedigrees that were known to be non-identifiable, namely the half-cousin-type and full-cousin-type relationships.  It also introduces a troubling new pair of non-identifiable pedigrees that are also half-cousin-type pedigrees but which contain three individuals of interest.  This is the first discussion of non-identifiable pedigrees with genetic data available for more than two individuals, demonstrating that identifiability is not restricted to pedigrees having two individuals with data.

We introduce a general criteria that can be used to detect non-identifiable pedigrees.  We show how non-identifiable likelihoods relate to non-identifiable kinship coefficients.  An example is given showing that the kinship coefficient can be identical while the likelihood is sufficient to distinguish the pedigrees.  Finally, we show that a broad class of pedigree pairs, namely those with different numbers of necessary edges, are identifiable, and the necessary edges can be obtained in polynomial time.  We also introduce a class of pedigrees, i.e.~diploid Wright-Fisher genealogies, which are provably identifiable.  

In order to effectively deal with non-identifiable pedigrees, we can use Bayes rule to obtain the posterior as a function of the likelihood and the prior.  Some mild conditions on the prior mean that the posterior will distinguish among the potential pedigrees.  The class of dWF pedigrees provides such a prior.  Furthermore, we could use Theorem~\ref{thm:cardinalityofedges} as a guide to integrate over the uncertainty we have about a pedigree structure.


\bibliography{pedigree}
\bibliographystyle{plain}

\newpage

\setcounter{theorem}{1}
\setcounter{lemma}{0}
\setcounter{corollary}{0}

\section*{Appendix}

\begin{corollary}
For unlinked sites $\theta_t = 0.5$ for all $1 \le t \le T-1$, for any pedigree graphs $P$ and $Q$ with maximum ensemble states $|m(P)| = |m(Q)|$ and identical identity states, the pedigrees are non-identifiable.
\end{corollary}
\begin{proof}
For $\theta_t = 0.5$ for all $1 \le t \le T-1$, for any pedigree graphs $P$ and $Q$ with maximum ensemble state $|m(P)| = |m(Q)|$ and identical identity states, the transition equality is satisfied for any $\phi$ that preserves the identity states.  If the identity states are identical, then emission equality is satisfied.  This is because the identity states are a sub-partition of the emission partition, and because the emission probabilities of the identity states must be identical.  Together this means that pedigree graphs $P$ and $Q$ are non-identifiable for unlinked sites $\theta_t = 0.5$. \qed
\end{proof}

 The \emph{kinship
  coefficient} for a pair of individuals of interest is defined as the
probability of IBD when randomly choosing one allele from each
individual of interest.  Let the two individuals of interest be $\chi
= \{a,b\}$.  We write the kinship coefficient for $\chi$ as
$\Phi_{I}(P)_{\chi} = \sum_{x} \frac{\eta(x,\chi)}{4} \frac{1}{2^n}$
where $\eta(x,\chi)$ is the number of pairs of alleles of interest
$\chi_0 \cup \chi_1$ sharing the same connected component in $R_x$ and
$\chi_0 \cup \chi_1 = \{\{a_0,b_0\}, \{a_0,b_1\}, \{a_1,b_0\},
\{a_1,b_1\}\}$.  

The kinship coefficient can be rewritten as an
expectation over the condensed identity states which for general $I$
is the emission partition.  Therefore
\vspace{-0.25cm}
\[\Phi_{I}(P)_{\chi} = \sum_{E_x(P)} \frac{\eta(x,\chi)}{4} \frac{|E_x|}{2^n}.
\vspace{-0.25cm}
\]

\begin{corollary}
For unlinked sites $\theta_t = 0.5$ for all $1 \le t \le T-1$, given two 
non-identifiable pedigree graphs, $P$ and $Q$, with two
individuals of interest $\chi=\{a,b\}$, the kinship coefficient is
identical.
\end{corollary}
\begin{proof}
Since $P,Q$ are non-identifiable, there is a proper isomorphism between $\psi:m(P) \to m(Q)$ which we can use to obtain an isomorphism $\gamma$ between the emission partitions of $P$ and $Q$ such that the coefficients in the kinship sums are equivalent.  The existence of $\gamma$ means that the kinship coefficients are identical.

To obtain $\gamma$, we simply take the isomorphism induced on $\{E_x(P)|~\forall x\} \to \{E_x(Q)|~\forall x\}$ by $\psi$.  Since the emission partition preserves the emission probabilities for all input data, and since these probabilities are a function of the connected components of $R_x$, the emission partition preserves $\eta(x,I)$, meaning that for all  $y\in E_x,~\forall x$, $\eta(y,I) = \eta(\gamma(y), I)$.  Since $\gamma$ preserves the emission partition, it also preserves the $\eta$ coefficients.  Therefore, this $\gamma$ proves that the kinship coefficients of $P$ and $Q$ are identical. \qed
\end{proof}

\begin{lemma}
We say that an edge is \emph{removed} if its bit is set to a fixed value in all the inheritance vectors.
Any superfluous edge can be removed without changing the value of the likelihood.
\end{lemma}
\begin{proof}
First consider $P$ and unnecessary edge $e$ with corresponding indicator vector $\sigma_e$.  Let $Q$ be the pedigree with edge $e$ removed where the edge is removed by fixing the value of bit $e$ to zero.
We use the notation $x_e$ to refer to the $e$th bit of inheritance vector $x$.  

We will prove that the likelihoods are the same by proving that there is a proper isomorphism $\psi$ from the states of $Y^P_t$, the original pedigree HMM to the states of the removed-edge HMM $Y^Q_t$.  Furthermore, we have the property that $x \in \mathcal{H}_P$ has one more bit than $\bar{x} \in \mathcal{H}_Q$.  This means that we need to prove that $\psi$ satisfies both the transition and emission equalities.

We will first note that the emission probabilities are the same if we
remove edge $e$ as can be seen by the second property of the
superfluous edge definition.  So any $\psi$ satisfies the emission
equality if it maps $x \in \mathcal{H}_P$ to $\sigma(x)$ if $x_e = 1$
and to $x$ if $x_e =0$.  For the rest of the proof, we will consider
such a $\psi$.

Now, we need only prove that one of these $\psi$ satisfying the emission equality also satisfies the transition equality.  We will do this by a short computation on the transition probabilities.
Notice that $\mathcal{H}_P$ is the union of two sets $S_1 = \{x | x_e = 1\}$ and $S_0 = \{x | x_e = 0\}$.
We can also write that $\forall x \in S_0$, $\sigma(x) \in S_1$.  
Recall that the transition probabilities are written, 
for $x \in W_i$ and for $i \ne j$, as 
\begin{eqnarray*}
Pr[Y^P_{t+1}=W_j~|~Y^P_t = W_i] &=& \sum_{y \in W_j} Pr[X_{t+1}=y~|~X_t=x] \\
    &=& \sum_{y \in W_j \cap S_0} Pr[X_{t+1}=y~|~X_t=x] \\
            && + \sum_{y \in W_j \cap S_1} Pr[X_{t+1}=y~|~X_t=x] \\
    &=& \sum_{y \in W_j \cap S_0} Pr[X_{t+1}=y~|~X_t=x] \\
            && + \sum_{y \in W_j \cap S_1} Pr[X_{t+1}=y~|~X_t=\sigma(x)]. 
\end{eqnarray*}
We can make the last statement due to edge $e$ not influencing the emission of the HMM, by the second property of the definition of a superfluous edge.  This is because the edge $e$ must not be on any direct path connecting two individuals of interest.  Therefore, we can conclude that  $x \in W_i $ and $\sigma(x) \in W_i$.  Furthermore without loss of generality, we will assume that $x \in W_i \cap S_0$.

Continuing on, we can finish the proof by employing the first property of the definition of a superfluous edge to get that
\begin{eqnarray*}
Pr[Y^P_{t+1}=W_j~|~Y^P_t = W_i]
    &=& 2 \sum_{y \in W_j \cap S_0} Pr[X_{t+1}=y~|~X_t=x] \\
 \label{e:trans1}   &=& 2  \sum_{y \in W_j \cap S_0} \theta_t^{H(x,y)}(1-\theta_t)^{n-1-H(x,y)} (1-\theta_t) \\
 \label{e:trans2}   &=& 2  (1-\theta_t) \sum_{y \in W_j \cap S_0} \theta_t^{H(x,y)}(1-\theta_t)^{n-1-H(x,y)} \\
 \label{e:trans3}   &=& \sum_{\bar{y} \in \psi(W_j)} \theta_t^{H(\bar{x},\bar{y})}(1-\theta_t)^{n-1-H(\bar{x},\bar{y})} \\
    &=& Pr[Y^Q_{t+1}=\psi(W_j)~|~Y^Q_t = \psi(W_i)]
\end{eqnarray*}
where the second and third lines are due to the definition of the transition probability.  
We note that removing an edge by fixing its bit-value in all the inheritance vectors is nearly equivalent to removing the edge's bit entirely. 
The easiest way to see the equality of the last few lines is to note that the transition probabilities are distributions---they must sum to one---and therefore proportionality implies equality.
So, we are able to conclude that there is a $\psi$ satisfying the emission equality and the transition equality.
\qed
\end{proof}

\begin{theorem}
If pedigrees $P$ and $Q$ have a different number of \emph{necessary} edges, then there is no proper isomorphism and the likelihoods for $P$ and $Q$ are identifiable.
\end{theorem}
\begin{proof}
We will prove this using the contra-positive.  Suppose that $P$ and $Q$ are not identifiable. Then there is a proper isomorphism $\psi:m(P) \to m(Q)$.  By removing unnecessary edges from $P$ and $Q$ we will show that they have the same number of necessary edges proving the statement.

Call $P'$ the pedigree with $E(P') = E(P) \setminus \{e\}$.  By the sequence of equalities above, we have that $P'$ and $Q$ have a proper isomorphism since the transition equalities were maintained and the emission equality is unchanged. We check each edge of $P$ and $Q$ removing any edges that are unnecessary to obtain $P'$ and $Q'$ that contain only necessary edges.  If a pedigree has an unnecessary edge that would be evident by comparing the polynomials in the transition probabilities and seeing that there are twice the number of terms with the same powers. The sequence of edge removals yields a sequence of proper isomorphisms which means that the projection of $\psi$ onto the unremoved edges is a proper isomorphism for $P'$ and $Q'$.  Remove edges until there are no superfluous edges in $P'$ or $Q'$.  

Once the superfluous edges have been removed, $|E(P')| = |E(Q')|$, since the map $\psi$ on the polynomials of the transition probabilities guarantee a one-to-one correspondence between the terms of $P'$ and the terms of $Q'$, since there must be the same number of like-powers.  This ensures that the number of inheritance vectors and therefore the number of edges are equal. \qed
\end{proof}

The maximum ensemble partition consists of a
group of isometries acting on the state-space $\mathcal{H}_P$.  An
isometry is any function $T$ such that $|T(x) \oplus T(y)| = |x \oplus
y|$ for all $y \in W_i$ and $x \in W_j$, for all $i$ and $j$.
This means that the transition probabilities satisfy 
\[
Pr[X_{t+1}=y |X_t = x] = Pr[X_{t+1} = T(y) | X_t = T(x)].
\]
For details see~\cite{Kirkpatrick2011xx}.

\begin{lemma}
An individual is \emph{superfluous} if and only if all the edges adjacent to that individuals are \emph{superfluous}.
\end{lemma}
\begin{proof}
($\Rightarrow$)  Take condition (1).  If $i$ is not an ancestor of any individual of interest, then it is a descendant of some ancestor of $\chi$.  Consider a parent edge $e$.  Since individual $i$ has no data, we have $Pr[G_t| X_t = x] = Pr[G_t| X_t = \sigma_e \oplus x]$. Since this holds, both $x$ and $\sigma_e \oplus x$ are in the same set of the emission partition.  The same holds for any $y$ and $\sigma_e \oplus y$.
Therefore, we can apply the isomorphism $\sigma_e$ to the transition probabilities using the Markov property, and obtain
$Pr[X_{t+1} = y | X_t = x] = Pr[X_{t+1} = \sigma \oplus y | X_t = \sigma \oplus x]$.  This is because any isometry that is an element of the maximal isomoetry group, having the orbits $\{W_1,...,W_k\}$, of the compressed Markov chain will certainly satisfy the Markov property.

Consider condition (2).  Now $i$  is an ancestor of some individuals in $\chi$, but the lineages in $\chi$ have coalesced before reaching $i$.  Again $i$ is an individual without data.  So, we can apply the same argument as for condition (1).

($\Leftarrow$)  Assume that $i$ does not satisfy either of the conditions for being superfluous.  Then it is on some simple undirected path connecting two nodes of interest $a$ and $b$ where a path is simple if no edge is repeated.  Then there exists $x$, given by the simple path, such that $Pr[G_{t}| X_t = x] \ne Pr[G_{t} | X_t = \sigma \oplus x]$ and condition (2) of the definition of a superfluous edge is violated. \qed
\end{proof}

\begin{theorem}
Two non-isomorphic, diploid Wright-Fisher pedigrees $P$ and $Q$ contain only necessary edges and have individuals of interest $\chi$ labeling the `leafs' which are the individuals with no children.
Then pedigrees $P$ and $Q$ are identifiable.
\end{theorem}
\begin{proof}
Since pedigrees $P$ and $Q$ are not isomorphic, only \emph{maximal partial isomorphisms} $\alpha:I(P) \to I(Q)$ can be created.  We consider a partial isomorphism one that maps a connected subgraph of $P$ with nodes $U$ to a connected subgraph of $Q$ with nodes $V$, where the nodes $\chi \in U$ and $\chi \in V$, exist in both subgraphs.  A maximal partial isomorphism is one where no further nodes $i \in I(P)$ and $j \in I(Q)$  can be paired while maintaining the following property on the edges: $u,v \in U$ such that $(u,v) \in E(P)$ implies that $(\alpha(u), \alpha(v)) \in E(Q)$.

For each partial isomorphism $\alpha$, there exists an edge $e \in E(P)$ which is not mapped to $I(Q)$, because $P$ and $Q$ are not isomorphic.
Since $e$ is a necessary edge and $P$ and $Q$ are leaf-labeled, it lies on some simple path connecting individuals $a \in \chi$ and $b \in \chi$ where $a \ne b$.  Notice that $a$ and $b$ cannot be connected in $Q$ via a path of the same length as in $P$, otherwise $\alpha$ would not be maximal as there is a matching edge for $e$ in the path.

Let the \emph{most recent common ancestor (MRCA)} of $a$ and $b$ be the youngest individual who is an ancestor of both $a$ and $b$.  
Let $\Pi$ be a pedigree, then $M_{ab}(\Pi)$ is the number of edges on the path between $a$, the MRCA of $a$ and $b$, and $b$ in pedigree $\Pi$.  Because of the Wright-Fisher assumption, there are three cases for the path in $Q$:
\begin{enumerate}
\item $a$ and $b$ are not connect in $Q$,
\item $a$ and $b$ are connected in $Q$ via a path having MRCA $M_{ab}(Q) < M_{ab}(P)$, and
\item $a$ and $b$ are connected in $Q$ via a path having MRCA $M_{ab}(Q) > M_{ab}(P)$.
\end{enumerate}
Without the Wright-Fisher assumption, it would be possible for $M_{ab}(Q) = M_{ab}(P)$.
In all three cases the emission probability for the path in $Q$ is different for the emission probability of the path in $P$.  
This is because, for any such paths in $P$ and $Q$ as detailed above, any hidden state $x$ in the state-space which contains the path in $P$ as a subgraph of the $R_x$ will have a different emission probability for some data than a hidden state $y$ which contains the path in $Q$.  Furthermore, we know that since $\alpha$ did not produce a full isomorphism, there is no other element $y'$ in the state-space of $Q$ that has the same emission probability as $x$.
This proves that $P$ and $Q$ have distinct likelihoods and are identifiable. \qed
\end{proof}

\end{document}